\documentclass[journal]{IEEEtran}
\usepackage{cite}
\usepackage{amsmath,amssymb,amsfonts}
\usepackage{algorithmic}
\usepackage{graphicx}
\usepackage{textcomp}
\usepackage{makecell}
\usepackage[hidelinks]{hyperref}

\usepackage{bm}
\usepackage{extarrows}
\usepackage{amsthm}
\usepackage{color}
\usepackage{enumerate}
\usepackage{bookmark}
\usepackage{mathtools}
\graphicspath{{figure}}
\usepackage{setspace}
\usepackage{multirow}
\usepackage{subfigure}
\usepackage{float}
\usepackage{booktabs}
\usepackage{algorithm}
\usepackage{nicematrix}
\usepackage{cite}
\usepackage{framed} 
\usepackage{scalerel}
\usepackage{tikz}
\usetikzlibrary{matrix}

\newcommand{\BE}{\begin{equation}}
\newcommand{\EE}{\end{equation}}
\newcommand{\BS}{\begin{subequations}}
\newcommand{\ES}{\end{subequations}}
\newtheorem{theorem}{Theorem}
\newtheorem{proposition}{Proposition}

\newtheorem{definition}{Definition}

\newtheorem{remark}{Remark}
\newtheorem{lemma}{Lemma}
\newtheorem{corollary}{Corollary}

\begin{document}
 
\title{Algebra of L-banded Matrices}

\author{Shunqi Huang, \IEEEmembership{Graduate Student Member, IEEE}, Lei Liu, \IEEEmembership{Senior Member, IEEE} and Brian~M.~Kurkoski, \IEEEmembership{Member, IEEE}
\thanks{Shunqi Huang and Brian~M.~Kurkoski are with the School of Information Science, Japan Institute of Science and Technology (JAIST), Nomi 923-1292, Japan (e-mail: \{shunqi.huang, kurkoski\}@jaist.ac.jp). Lei~Liu is with the Zhejiang Provincial Key Laboratory of Information Processing, Communication and Networking, College of Information Science and Electronic Engineering, Zhejiang University, Hangzhou 310007, China (e-mail: lei\_liu@zju.edu.cn).  \emph{(Corresponding author: Lei Liu.)}} 
}

\maketitle

\begin{abstract}
Convergence is a crucial issue in iterative algorithms. Damping is commonly employed to ensure the convergence of iterative algorithms. The conventional ways of damping are scalar-wise, and either heuristic or empirical. Recently, an analytically optimized vector damping was proposed for memory message-passing (iterative) algorithms. As a result, it yields a special class of covariance matrices called L-banded matrices. In this paper, we show these matrices have broad algebraic properties arising from their L-banded structure. In particular, compact analytic expressions for the LDL decomposition, the Cholesky decomposition, the determinant after a column substitution, minors, and cofactors are derived. Furthermore, necessary and sufficient conditions for an L-banded matrix to be definite, a recurrence to obtain the characteristic polynomial, and some other properties are given. In addition, we give new derivations of the determinant and the inverse.

\end{abstract}


\IEEEpeerreviewmaketitle

\section{Introduction}
\IEEEPARstart{V}ARIOUS kinds of iterative algorithms are widely used in the fields of statistical signal processing, compressed sensing, communications, machine learning, coding theory, etc. For example, gradient descent algorithms \cite{ruder2016overview, saad2003iterative} are used for convex optimization, Jacobi and Gauss-Seidel algorithms \cite{bulirsch2002introduction, burden2015numerical} for linear systems, and message passing algorithms \cite{donoho2009message, bayati2011dynamics, ma2017orthogonal, rangan2019vector, takeuchi2021bayes, murphy1999loopy, liu2022memory} for graphical models. 

An iterative algorithm can be generally represented as
\begin{align}
    x_t = f_t(x_1, \cdots\!, x_{t-1}), \label{Eqn:I1}
\end{align}
where $x_t$ is the $t$-th estimate for $x$. A non-memory iterative algorithm $x_t = f_t(x_{t-1})$ can be seen as a special case of (\ref{Eqn:I1}). How to guarantee convergence of iterative algorithms is a crucial problem. Damping is an efficient technique to ensure the convergence of iterative algorithms. \emph{Scalar} damping \cite{murphy1999loopy} 
\begin{align}
    \hat{x}_t = (1-\zeta)\hat{x}_{t-1} + \zeta x_t
\end{align}
is commonly used in the existing literature, where $\zeta$ is a scalar damping factor determined empirically, and $\hat{x}_t$ is the damped estimate initialized with $\hat{x}_1 = x_1$.

Recently, a non-empirical and analytically optimized \emph{vector} damping was proposed for memory approximate message passing (MAMP) algorithms \cite{liu2022memory} as 
\begin{align}
    \hat{x}_t   = [x_1, \cdots\!, x_t]\, \bm{\zeta}_{t} \quad {\rm and} \quad
    \bm{\zeta}_t  =  \tfrac{\bm{V}_{t}^{-1}\bm{1}}{\bm{1}^{\rm T} \bm{V}_{t}^{-1}\bm{1}},
\end{align} 
where $\bm{\zeta}_{t}$ is a damping vector, $\bm{V}_t$ the covariance matrix of $x_1, \cdots\!, x_t$, $\bm{1}$ the all-one vector, $\bm{V}_{t}^{-1}\bm{1}$ the column-wise sum of $\bm{V}_{t}^{-1}$, and $\bm{1}^{\rm T}\bm{V}_{t}^{-1}\bm{1}$  the sum of all the entries in $\bm{V}_{t}^{-1}$. It was found in \cite{liu2022memory} that the covariance matrix $\hat{\bm{V}}_t$ of the damped estimates $\hat{x}_1, \cdots\!, \hat{x}_t$ has a special structure that the entries in each ``L band'' are identical. Meanwhile, this special structure was found independently in \cite{takeuchi2022convergence}, and played an important role in the convergence proof of orthogonal/vector AMP\cite{ma2017orthogonal,rangan2019vector}. This particular type of matrix is referred to as L-banded matrices as follows.

\begin{definition}[L-banded Matrix \cite{liu2022memory, takeuchi2022convergence}]\label{Def:L} 
A matrix $\bm{A} \equiv [a_{i,j}] \in \mathbb{R}^{n \times n}$ is said to be an L-banded matrix if
\BS\label{Eqn:L-banded}
\begin{align}
    a_{i, j} = a_{t, t},\ t = \max(i, j).
\end{align}
That is, $\bm{A}$ can be written as
\begin{align}
\bm{A} =
\begin{tikzpicture}[baseline=-\dimexpr-.0mm\relax]
  \matrix [matrix of math nodes,left delimiter={[},
 right delimiter={]}, row sep=2.5mm,column sep=2.5mm, ampersand replacement=\&] (M) 
 {
   a_1 \& a_2 \& \cdots \& a_n          \\ 
   a_2 \& a_2 \& \cdots \& a_n          \\ 
   \vdots \& \vdots \& \ddots \& \vdots    \\
   a_n \& a_n \& \cdots \& a_n             \\ 
   };
   \draw[](M-1-2.north east)--(M-2-2.south east)--(M-2-1.south west)--(M-2-1.north west)--(M-2-2.north west)--(M-1-2.north west)--(M-1-2.north east);
   \draw[](M-1-4.north east)--(M-4-4.south east)--(M-4-1.south west)--(M-4-1.north west)--(M-4-4.north west)--(M-1-4.north west)--(M-1-4.north east);
\end{tikzpicture}.
\end{align}
\ES
\end{definition}
It is crucial to emphasize that the matrices satisfying (\ref{Eqn:L-banded}) are referred to as \emph{L-matrices} when $a_i \in \mathbb{C}$ and $n \to \infty$ in \cite{bouthat2021matrices, bouthat2021norm, bouthat2022critical, vstampach2022hilbert, vstampach2022asymptotic}. In other words, L-banded matrices can be regarded as L-matrices with real and finite entries.  

Significantly, it was shown that any L-banded covariance matrix converges in \cite{liu2022memory, liu2021sufficient, liu2022sufficient}, i.e., its diagonal entries converge (See Lemma 2). As a result, the convergence difficulties of general iterative algorithms were resolved in principle. The discussions in \cite{liu2022memory, liu2021sufficient, liu2022sufficient, takeuchi2022convergence} are limited to L-banded covariance matrices. Although some properties of L-banded matrices were included in previous works (see Section \ref{Section:PR}), more algebraic properties remain unknown. 

In this paper, we show that the structure of L-banded matrices gives rise to a number of algebraic properties, which are broader than those properties which have been uncovered thus far. This algebra of L-banded matrices, besides being potentially useful for signal processing applications, is also inherently interesting. 
These new properties include an analytic expressions for the LDL decomposition, the Cholesky decomposition, minors, cofactors, and the determinant of the matrix formed by replacing any one column. In addition, necessary and sufficient conditions for an L-banded matrix to be definite, a recurrence to obtain the characteristic polynomial, and some other properties are given. We also provide new derivations of the determinant and the inverse. Finally, a comparison of the time complexity of some operations on L-banded matrices and those of general matrices is given.

\subsection{Notation}
Boldface lowercase and boldface uppercase symbols denote column vectors and matrices, respectively. $\bm{A} \equiv [a_{i,j}]_{n \times n}$ denotes that $\bm{A}$ is an $n \times n$ matrix with $(i,j)$-th entries $a_{i,j}$. $\bm{I}_n$ denotes an $n \times n$ identity matrix. ``iff'' is the abbreviation of ``if and only if''. We call the matrix in (\ref{Eqn:L-banded}) an $n \times n$ L-banded matrix with $[a_1, \cdots\!, a_n]$. In the following paragraphs, we assume that $\bm{A}$ is an $n \times n$ L-banded matrix with $[a_1, \cdots\!, a_n]$ unless otherwise specified.

\subsection{Previous Results}\label{Section:PR}
We summarize the properties of L-banded matrices proposed in previous works \cite{liu2022memory, liu2021sufficient, vstampach2022hilbert, liu2022sufficient, takeuchi2022convergence} as follows.
\begin{lemma}[Convergence of Positive Definite $\bm{A}$\cite{takeuchi2022convergence}]\label{lemma:pos}
    If $\bm{A}$ is positive definite, then $a_1 > a_2 > \cdots > a_n > 0$.
\end{lemma}
\begin{lemma}[Convergence of Positive Semi-definite $\bm{A}$\cite{liu2022memory}]\label{lemma:semi} 
    If $\bm{A}$ is a covariance (i.e., positive semi-definite) matrix, then $a_1 \geq a_2 \geq \cdots \geq a_n \geq 0$.
\end{lemma}
\begin{lemma}[Positive Semi-definite $\bm{A}$ \cite{vstampach2022hilbert}]\label{lemma:semi2}
    $\bm{A}$ is positive semi-definite iff $a_1 \geq a_2 \geq \cdots \geq a_n \geq 0$.
\end{lemma}

\begin{lemma}[Determinant\cite{takeuchi2022convergence, vstampach2022hilbert}]\label{lemma:det}
The determinant of $\bm{A}$ is
\begin{align}
    |\bm{A}| = a_n\prod_{k=1}^{n-1}(a_k - a_{k+1}). 
\end{align}
\end{lemma}

\begin{corollary}[Invertibility\cite{vstampach2022hilbert, liu2021sufficient, liu2022sufficient}]
$\bm{A}$ is invertible iff $a_n \neq 0$ and $a_k \neq a_{k+1}, \forall k \in \{1, \cdots\!, n\!-\!1\}$.
\end{corollary}

\begin{lemma}[Inverse\cite{liu2021sufficient, liu2022sufficient, vstampach2022hilbert}]\label{lemma:inv}
Suppose that $\bm{A}$ is invertible. Let $\delta_k = (a_k - a_{k+1})^{-1}$ for $k \in \{1, \cdots\!, n\!-\!1\}$, $\delta_0 = 0$ and $\delta_n = a_n^{-1}$. The inverse $\bm{A}^{-1} \equiv [b_{i,j}]_{n \times n}$ is given by 
\BS\label{Eqn:inv_1}
\begin{align}
    b_{i, j} = 
    \begin{cases}
        \delta_{i-1} + \delta_i, & \quad i = j \\[1mm]
        -\delta_{\min(i,j)}, & \quad |i-j| = 1 \\[1mm]
        0, & \quad |i-j| > 1
    \end{cases}.
\end{align}
That is, $\bm{A}^{-1}$ can be written as
\begin{align}
    \bm{A}^{-1} \!\!=\!\!
    \begin{bmatrix} 
        \delta_1 \!\!\! & \!\!\! -\delta_1 \!\!\! & \!\!\! & \!\!\! & \!\!\! \\
        -\delta_1 \!\!\! & \!\!\! \delta_1\!+\!\delta_2 \!\!\! & \!\!\! -\delta_2 \!\!\! & \!\!\! & \!\!\! \text{\Large 0} \\
        \!\!\! & \!\!\! \ddots \!\!\! & \ \ddots \!\!\! & \!\!\! \ddots \!\!\! & \!\!\!  \\
        \text{\Large 0} \!\!\! & \!\!\! & \!\!\! -\delta_{n-2} \!\!\! & \!\!\! \delta_{n-2}\!+\!\delta_{n-1} \!\!\! & \!\!\! -\delta_{n-1}  \\
        \!\!\! & \!\!\! & \!\!\! & \!\!\! -\delta_{n-1} \!\!\! & \!\!\! \delta_{n-1}\!+\!\delta_{n}
    \end{bmatrix}.
\end{align}
\ES
\end{lemma}

\begin{section}{Main Results}
In this section, we give our main results. Necessary and sufficient conditions for the definiteness of $\bm{A}$ are given in Section \ref{section:def}. The LDL decomposition and the Cholesky decomposition are given in Section \ref{section:LDL}. The minors, cofactors, and determinant after a column substitution are given in Section \ref{section:minor}. The characteristic polynomial is given in Section \ref{section:poly}. Some other properties are given in Section \ref{section:other}. The new proofs of Lemma \ref{lemma:det} and Lemma \ref{lemma:inv} are given in Section \ref{section:det} and \ref{section:inv}, respectively.

\subsection{Definiteness} \label{section:def}
\begin{lemma}\label{lemma:qua}
For $\bm{x} = [x_1, \cdots\!, x_n]^{\rm T} \in \mathbb{R}^{n}$, 
\BS
\begin{align}
    \bm{x}^{\rm T} \bm{A} \bm{x} &= \sum_{k=1}^n \Delta_k \Big(\sum_{i=1}^k x_i\Big)^2,
\end{align}
where $\Delta_k = a_k - a_{k+1}$ for $k \in \{1, \cdots\!, n\!-\!1\}$ and $\Delta_n = a_n$.
\ES
\end{lemma}
\begin{proof}
With $\bm{x} = [x_1, \cdots\!, x_n]^{\rm T}$,
\BS
\begin{align}
    &\bm{x}^{\rm T} \bm{A} \bm{x}
    = \sum_{k=1}^n x_k \Big(a_k \sum_{i=1}^{k} x_i + \sum_{i=k+1}^n a_i x_i\Big) \\
    &= \sum_{k=1}^n a_k x_k^2 + \sum_{k=1}^n \Big(a_k \sum_{i=1}^{k-1} x_i x_k + \sum_{i=k+1}^{n} a_i x_k x_i\Big) \\
    &= \underbrace{\sum_{k=1}^n a_k x_k^2}_{(\star)} + \underbrace{2 \sum_{k=2}^n a_k \sum_{i=1}^{k-1} x_i x_k}_{(\#)}. \label{Eqn:def1}
\end{align}
\ES
$(\star)$ can be expanded as
\BS
\begin{align}
    (\star) &= a_n \sum_{k=1}^n x_k^2 + \sum_{k=1}^{n-1}(a_k - a_n) x_k^2 \\
        &= a_n \sum_{k=1}^n x_k^2 + \sum_{k=1}^{n-1}(a_k - a_{k+1}) \sum_{i=1}^k x_i^2 \\
        &= \sum_{k=1}^n \Delta_k \sum_{i=1}^k x_i^2.
\end{align}
\ES
$(\#)$ can be expanded as
\BS
\begin{align}
    &(\#) = 2 \sum_{k=2}^n a_k \Big(\sum_{p=2}^k\sum_{i=1}^{p-1}x_i x_p - \sum_{p=2}^{k-1}\sum_{i=1}^{p-1}x_i x_p \Big)\\
    &= 2 \Big(a_n\sum_{p=2}^n\sum_{i=1}^{p-1}x_i x_p + \sum_{k=2}^{n-1}(a_k - a_{k+1})\sum_{p=2}^k\sum_{i=1}^{p-1} x_i x_p \Big) \\
    &= 2 \sum_{k=2}^{n}\Delta_k\sum_{p=2}^k\sum_{i=1}^{p-1} x_i x_p.
\end{align}
\ES
Then, (\ref{Eqn:def1}) can be rewritten as
\begin{align}
    \bm{x}^{\rm T} \bm{A} \bm{x} = \sum_{k=1}^n \Delta_k \Big(\sum_{i=1}^k x_i\Big)^2.
\end{align}
Thus, we finish the proof.
\end{proof}

Lemma \ref{lemma:qua} not only shows a simpler expression for the quadratic form, but also is the key to proving Theorem \ref{theorem:def}.

\begin{theorem}[Definiteness]\label{theorem:def}
The following statements hold:
\begin{enumerate}
\item{$\bm{A}$ is positive definite iff $a_1 > a_2 > \cdots > a_n > 0$.}
\item{$\bm{A}$ is positive semi-definite iff $a_1 \geq a_2 \geq \cdots \geq a_n \geq 0$.}
\item{$\bm{A}$ is negative definite iff $a_1 < a_2 < \cdots < a_n < 0$.}
\item{$\bm{A}$ is negative semi-definite iff $a_1 \leq a_2 \leq \cdots \leq a_n \leq 0$.}
\end{enumerate}
\end{theorem}
\begin{proof}
$\bm{A}$ is positive definite iff 
\begin{align}
    \bm{x}^{\rm T} \bm{A} \bm{x} > 0,\ \forall \bm{x} \in \mathbb{R}^{n}\backslash\{\bm{0}\}. \label{Eqn:the3_1}
\end{align}
From Lemma \ref{lemma:qua}, (\ref{Eqn:the3_1}) is equivalent to: For any real $x_1, \cdots\!, x_n$ that are not all zeros, 
\begin{align}
    \sum_{k=1}^n \Delta_k \Big(\sum_{i=1}^k x_i\Big)^2 > 0. \label{Eqn:the3_2}
\end{align}
Note that $(\sum_{i=1}^1 x_i)^2, \cdots\!, (\sum_{i=1}^n x_i)^2$ are not all zeros since $x_1, \cdots, x_n$ are not all zeros. Thus, (\ref{Eqn:the3_2}) is equivalent to
\begin{align}
    \Delta_k > 0,\ \forall k \in \{1, \cdots\!, n\}. \label{Eqn:the3_3}
\end{align}
(\ref{Eqn:the3_3}) is equivalent to 
\begin{align}
    a_1 > a_2 > \cdots > a_n > 0.
\end{align}
The proofs for the other cases are omitted since they are similar.
\end{proof}

Compared to Lemma \ref{lemma:pos} and Lemma \ref{lemma:semi}, Theorem \ref{theorem:def} shows that the necessary conditions for $\bm{A}$ to be positive definite and positive semi-definite are also sufficient conditions. Though the second statement was proved in Lemma \ref{lemma:semi2}, Theorem \ref{theorem:def} gives a different proof. In addition, necessary and sufficient conditions for $\bm{A}$ to be negative definite and negative semi-definite are listed in Theorem \ref{theorem:def}.

\subsection{LDL decomposition and Cholesky decomposition}\label{section:LDL}

The LDL decomposition of a real symmetric matrix has the form $\bm{L}\bm{D}\bm{L}^{\rm T}$, where $\bm{L}$ is a unit lower triangular matrix and $\bm{D}$ is a diagonal matrix.

\begin{theorem}[LDL decomposition]\label{theorem:LDL}
Suppose that $a_k \neq 0, \forall k \in \{1, \cdots\!, n\!-\!1\}$. The LDL decomposition of $\bm{A}$ can be
\BS
\begin{align}
    \bm{A} = \bm{L}\bm{D}\bm{L}^{\rm T},
\end{align}
where $\bm{L} \equiv [l_{i,j}]_{n \times n}$ is a unit lower triangular matrix as
\begin{align}
    l_{i, j} = 
    \begin{dcases}
        \frac{a_i}{a_j}, & \quad i > j \\
        1, & \quad i = j
    \end{dcases},
\end{align}
and $\bm{D}$ is a diagonal matrix with diagonal entries 
\begin{align}
    d_{k} = 
    \begin{dcases}
        a_1, & \quad k = 1 \\[1mm]
        \frac{a_k}{a_{k-1}}(a_{k-1}-a_k), & \quad 1 < k \leq n
    \end{dcases} \label{Eqn:d}.
\end{align}
That is, $\bm{L}$ and $\bm{D}$ can be written as
\begin{align}
    \bm{L} &= 
    \begin{bmatrix}
        1 &  &  &  & \\
        \frac{a_2}{a_1} & 1 &  & \text{\Large{0}} & \\
        \frac{a_3}{a_1} & \frac{a_3}{a_2} & \ddots &  & \\
        \vdots & \vdots & \vdots & 1 & \\[1mm]
        \frac{a_n}{a_1} & \frac{a_n}{a_2} & \cdots & \frac{a_n}{a_{n-1}} & 1
    \end{bmatrix}, \\[1mm]
    \bm{D} &= 
    \begin{bmatrix}
        a_1 &  &  & \qquad \text{\Large{0}} \\
        & \frac{a_2}{a_1}(a_1\!-\! a_2) &  & \\
        &  & \ddots & \\ 
        \text{\Large{0}} &  &  &  \frac{a_n}{a_{n-1}}(a_{n-1}\!-\!a_n)
    \end{bmatrix}.
\end{align}
\ES
\end{theorem}

\begin{proof}
    \BS
    For $i,j \in \{1, \cdots\!, n\}$, let $\bm{b}_i^{\rm T}$ denote the $i$-th row of $\bm{L}\bm{D}$ and $\bm{l}_j$ denote the $j$-th column of $\bm{L}^{\rm T}$. We have
    \begin{align}
        \bm{b}_i^{\rm T} &= \big[d_{1}\tfrac{a_i}{a_1}, \cdots\!, d_i\tfrac{a_i}{a_i}, \underbrace{0, \cdots\!, 0}_{n-i}\big], \\
        \bm{l}_j &= \big[\tfrac{a_j}{a_1}, \cdots\!, \tfrac{a_j}{a_j}, \underbrace{0, \cdots\!, 0}_{n-j}\big]^{\rm T}.
    \end{align}
    Let $m = \min(i, j)$, $t = \max(i, j)$, and $\bm{L}\bm{D}\bm{L}^{\rm T} \equiv [c_{i, j}]_{n \times n}$.
    \begin{align}
        c_{i, j} &= \bm{b}_i^{\rm T} \bm{l}_j = \sum_{k=1}^{m} d_k \frac{a_i a_j}{a_k^2} \\
        &= a_i a_j \Big(\frac{1}{a_1} + \sum_{k=2}^m \frac{a_{k-1}-a_k}{a_{k-1}a_k}\Big)  \\
        &= a_i a_j \frac{1}{a_m} \\
        &= a_t.
    \end{align}
    \ES
    Thus, we have proved that $\bm{A} = \bm{L}\bm{D}\bm{L}^{\rm T}$.
\end{proof}

\begin{proposition}
Let $p \in \{1, \cdots\!, n\}$ such that $a_p \neq 0$ and $a_{p+1} = \cdots = a_n = 0$. The LDL decomposition of $\bm{A}$ exists iff $a_k \neq 0,\forall k \in \{1, \cdots\!, p\!-\!1\}$.
\end{proposition}
\begin{proof}
First, we want to show ``the LDL decomposition of $\bm{A}$ exists if $a_k \neq 0, \forall k \in \{1, \cdots\!, p\!-\!1\}$''. Let $\Tilde{\bm{A}}$ be an L-banded matrix with $[a_1, \cdots\!, a_p]$. Since $a_1, \cdots\!, a_{p-1}$ are non-zero, we can let $\Tilde{\bm{A}} = \Tilde{\bm{L}}\Tilde{\bm{D}}\Tilde{\bm{L}}^{\rm T}$ by Theorem \ref{theorem:LDL}. Then, the LDL decomposition of $\bm{A}$ can be $\bm{A} = \bm{L}\bm{D}\bm{L}^{\rm T}$, where
\begin{align}
    \bm{L} = 
    \begin{bmatrix}
        \Tilde{\bm{L}} & \bm{0} \\
        \bm{0} & \bm{I}_{n-p} 
    \end{bmatrix},\ 
    \bm{D} = 
    \begin{bmatrix}
        \Tilde{\bm{D}} & \bm{0} \\
        \bm{0} & \bm{0} 
    \end{bmatrix}. 
\end{align}

Second, we want to show ``... only if ...''. We consider showing the contrapositive, i.e., the LDL decomposition of $\bm{A}$ does not exist if $\exists k \in \{1, \cdots\!, p\!-\!1\}, a_k = 0$, and use proof by contradiction. Assume that there exists an LDL decomposition $\bm{A}=\bm{L}\bm{D}\bm{L}^{\rm T}$ if $\exists k \in \{1, \cdots\!, p\!-\!1\}, a_k = 0$. We can always find $m \in \{1, \cdots\!, p\!-\!1\}$ such that $a_m = 0$ and $a_{m+1} \neq 0$. For $i,j \in \{1, \cdots\!, n\}$, the $i$-th row of $\bm{L}\bm{D}$ and the $j$-th column of $\bm{L}^{\rm T}$ are 
\BS
\begin{align}
    \bm{b}_i^{\rm T} &= \big[d_{1}l_{i, 1}, \cdots\!, d_{i-1}l_{i, i-1}, d_i, \underbrace{0, \cdots\!, 0}_{n-i}\big], \\
    \bm{l}_j &= \big[l_{j, 1}, \cdots\!, l_{j, j-1}, 1, \underbrace{0, \cdots\!, 0}_{n-j}\big]^{\rm T}.
\end{align}
\ES
Since $\forall j \leq m$, $a_m = \bm{b}_m^{\rm T}\bm{l}_j = 0$, we can get
\begin{align}
    l_{m, 1} d_1 = \cdots = l_{m, m-1} d_{m-1} = d_m = 0.
\end{align}
Then, we can obtain $a_{m+1} = \bm{b}_{m+1}^{\rm T}\bm{l}_m = 0$, which leads to a contradiction. Thus, we finish the proof.
\end{proof}

\begin{proposition}\label{proposition:unLDL}
There exists a unique LDL decomposition of $\bm{A}$ if $a_i \neq 0, \forall i \in \{1, \cdots\!, n\!-\!1\}$ and $a_{j-1} \neq a_j, \forall j \in \{2, \cdots\!, n\!-\!1\}$.
\end{proposition}
\begin{proof}
For $k \in \{1, \cdots\!, n\!-\!1\}$, the $k$-th leading principal submatrix of $\bm{A}$ is the matrix $\bm{A}$ with the last $n\!-\!k$ rows and columns removed, which is an L-banded matrix with $[a_1, \cdots\!, a_k]$. If $a_i \neq 0, \forall i \in \{1, \cdots\!, n\!-\!1\}$ and $a_{j-1} \neq a_j, \forall j \in \{2, \cdots\!, n\!-\!1\}$, all the leading principal submatrices of $\bm{A}$ are invertible. Thus, from Theorem 4.1.3 in \cite{golub2013matrix}, we can show that the LDL decomposition is unique.
\end{proof}

\begin{remark}
Suppose that $a_n \neq 0$. The sufficient condition of ``there exists a unique LDL decomposition of $\bm{A}$'' in Proposition \ref{proposition:unLDL} is also the necessary condition. 
\end{remark}

The Cholesky decomposition of a real symmetric positive definite matrix has the form $\bm{L}\bm{L}^{\rm T}$, where $\bm{L}$ is a lower triangular matrix.  

\begin{theorem}[Cholesky decomposition]\label{theorem:Cholesky}
Suppose that $\bm{A}$ is positive definite. The Cholesky decomposition of $\bm{A}$ is
\BS
\begin{align}
    \bm{A} = \hat{\bm{L}}\hat{\bm{L}}^{\rm T},
\end{align}
where $\hat{\bm{L}} \equiv [\hat{l}_{i,j}]_{n \times n}$ is a lower triangular matrix as
\begin{align}
    \hat{l}_{i,j} = \frac{a_i}{a_j}\sqrt{d_j}, & \quad i \geq j.
\end{align}
Recall
\begin{align}
    d_j = 
    \begin{dcases}
        a_1, & \quad j = 1 \\
        \frac{a_j}{a_{j-1}}(a_{j-1}-a_j), & \quad 1 < j \leq n
    \end{dcases}
    \nonumber
\end{align}
defined in (\ref{Eqn:d}).
That is, $\hat{\bm{L}}$ can be written as
\begin{align}
    \hat{\bm{L}} = 
    \begin{bmatrix}
        \sqrt{a_1} &  &  &  & \\
        \frac{a_2}{\sqrt{a_1}} & \sqrt{d_2} &  & \text{\Large{0}} & \\
        \frac{a_3}{\sqrt{a_1}} & \frac{a_3\sqrt{d_2}}{a_2} & \ddots &  & \\
        \vdots & \vdots & \vdots & \sqrt{d_{n-1}} & \\
        \frac{a_n}{\sqrt{a_1}} & \frac{a_n\sqrt{d_2}}{a_2} & \cdots & \frac{a_n\sqrt{d_{n-1}}}{a_{n-1}} & \sqrt{d_n}
    \end{bmatrix}.
\end{align}

\ES
\end{theorem}
\begin{proof}
    We have $a_1 > \cdots > a_n > 0$ since $\bm{A}$ is positive definite by Theorem \ref{theorem:def}. Then, from Theorem \ref{theorem:LDL}, the LDL decomposition of $\bm{A}$ is  
    \BS
    \begin{align}
        \bm{A} = \bm{L}\bm{D}\bm{L}^{\rm T}. \label{Eqn:LDL}
    \end{align}
    Note that the diagonal entries of $\bm{D}$ are all positive since $a_1 > \cdots > a_n > 0$. We can rewrite (\ref{Eqn:LDL}) as
    \begin{align}
        \bm{A} &= \bm{L}\bm{D}^{1/2}\bm{D}^{1/2}\bm{L}^{\rm T} \\
        &= \bm{L}\bm{D}^{1/2}\big(\bm{L}\bm{D}^{1/2}\big)^{\rm T}.
    \end{align}
    \ES
    Let $\hat{\bm{L}} = \bm{L}\bm{D}^{1/2}$, it is easy to prove the theorem.
\end{proof}

\subsection{Minors, Cofactors and Column Substitution}\label{section:minor}
 For $i,j \in \{1, \cdots\!, n\}$, let $M_{i,j}(\bm{A})$ be the $(i,j)$-th minor of $\bm{A}$, $C_{i,j}(\bm{A})$ be the $(i,j)$-th cofactor of $\bm{A}$, and $\bm{C}(\bm{A}) \equiv [C_{i,j}(\bm{A})]_{n \times n}$ be the cofactor matrix of $\bm{A}$.
\begin{theorem}[Minors and Cofactors]
    Suppose that $\bm{A}$ is invertible and $n \geq 2$. Then, the cofactors and minors of $\bm{A}$ are
    \BS
    \begin{align}
        C_{i,j}(\bm{A}) &= 
        \begin{cases}
            (\delta_{i-1}+\delta_i)|\bm{A}|, & \quad i = j \\[1mm]
            -\delta_{\min(i,j)}|\bm{A}|, & \quad |i-j| = 1 \\[1mm]
            0,  & \quad |i-j| > 1
        \end{cases}, \\[1mm]
        M_{i,j}(\bm{A}) &= (-1)^{i+j} C_{i,j}(\bm{A}).
    \end{align}
    \ES
    Recall $\delta_k = (a_k - a_{k+1})^{-1}$ for $k \in \{1, \cdots\!, n\!-\!1\}$, $\delta_0 = 0$ and $\delta_n = a_n^{-1}$ defined in Lemma \ref{lemma:inv}.
\end{theorem}
\begin{proof}
    Let ${\rm a}{\rm d}{\rm j}(\bm{A})$ be the adjugate matrix \cite{strang2006linear} of $\bm{A}$. Since $\bm{A}$ is symmetric,
    \begin{align}
        {\rm a}{\rm d}{\rm j}(\bm{A}) \equiv \bm{C}(\bm{A})^{\rm T} = \bm{C}(\bm{A}).
    \end{align}
    Since $\bm{A}$ is invertible, 
    \begin{align}
        {\rm a}{\rm d}{\rm j}(\bm{A}) = |\bm{A}|\bm{A}^{-1}.
    \end{align}
    Then, from Lemma \ref{lemma:inv}, we can prove the theorem.  
\end{proof}

 For $k \in \{1, \cdots\!, n\}$, let $\bm{A}\overset{k}{\leftarrow}\bm{b}$ denote the matrix formed by replacing column $k$ of $\bm{A}$ with $\bm{b}$.
\begin{theorem}[Determinant after Column Substitution]
    Suppose that $\bm{A}$ is invertible. Let $\bm{b} = [b_1, \cdots\!, b_n]^{\rm T} \in \mathbb{R}^n$. The determinant of $\bm{A}\overset{k}{\leftarrow}\bm{b}$ is given by
    \BS
    \begin{align}
        |\bm{A}\overset{k}{\leftarrow}\bm{b}| = g_k |\bm{A}|,
    \end{align}
    where
    \begin{align}
        g_k =
        \begin{cases}
            \delta_1 (b_1 - b_2), & \ k = 1 \\[1mm]
            \delta_{k-1}(b_k - b_{k-1}) + \delta_k(b_k - b_{k+1}), & \ 1 < k < n \\[1mm]
            \delta_{n-1}(b_n - b_{n-1}) + \delta_n b_n & \ k = n
        \end{cases}.
    \end{align}
    Recall $\delta_k = (a_k - a_{k+1})^{-1}$ for $k \in \{1, \cdots\!, n\!-\!1\}$ and $\delta_n = a_n^{-1}$ defined in Lemma \ref{lemma:inv}.
    \ES
\end{theorem}
\begin{proof}
    The Laplace expansion \cite{golub2013matrix} along the column $k$ of $|\bm{A}\overset{k}{\leftarrow}\bm{b}|$ is
    \begin{align}
        |\bm{A}\overset{k}{\leftarrow}\bm{b}| = \sum_{i=1}^n b_i C_{i,k}(\bm{A}).
    \end{align}
    Since ${\rm a}{\rm d}{\rm j}(\bm{A}) \equiv \bm{C}(\bm{A})^{\rm T}$,
    \begin{align}
        \big[|\bm{A}\overset{1}{\leftarrow}\bm{b}|, \cdots\!, |\bm{A}\overset{n}{\leftarrow}\bm{b}|\big]^{\rm T} = {\rm a}{\rm d}{\rm j}(\bm{A})\bm{b} = |\bm{A}|\bm{A}^{-1}\bm{b}.
    \end{align}
    Let $\bm{r}_k^{\rm T}$ be the $k$-th row of $\bm{A}^{-1}$. From Lemma \ref{lemma:inv},
    \BS
    \begin{align}
        \bm{r}_k^{\rm T} =
        \begin{cases}
            \big[\delta_1, -\delta_1, 0, \cdots\big], & \ k = 1 \\[1mm]
            \big[\underbrace{\cdots\!, 0}_{k-2}, -\delta_{i-1}, \delta_{i-1}\!+\!\delta_i, -\delta_i, 0, \cdots\big], & \ 1 < k < n \\
            \big[\cdots\!, 0, -\delta_{n-1}, \delta_{n-1}\!+\!\delta_n\big], & \ k = n
        \end{cases}
    \end{align}
    The $k$-th entry of $\bm{A}^{-1}\bm{b}$ is
    \begin{align}
        g_k = \bm{r}_k^{\rm T}\bm{b},
    \end{align}
    and then we can easily prove the theorem.
    \ES
\end{proof}

\subsection{Characteristic Polynomial}\label{section:poly}
\begin{theorem}[Characteristic Polynomial]
    Suppose that $\bm{A}$ is invertible. The characteristic polynomial of $\bm{A}$ is
    \BS
    \begin{align}
        p_A(\lambda) = |\bm{A}|f_n.
    \end{align}
     where $f_n$ can be obtained by a three-term recurrence: $f_0=1$, $f_1 = \delta_1\lambda-1$, and for $k \in \{2, \cdots\!, n\}$,
    \begin{align}\label{Eqn:rec}
        f_k = \big((\delta_{k-1}+\delta_k)\lambda-1\big)f_{k-1} - \delta_{k-1}^2\lambda^2 f_{k-2}.
    \end{align}
    Recall $\delta_k = (a_k - a_{k+1})^{-1}$ for $k \in \{1, \cdots\!, n\!-\!1\}$ and $\delta_n = a_n^{-1}$ defined in Lemma \ref{lemma:inv}. In other words, the eigenvalues of $\bm{A}$ are the roots of $f_n$.
    \ES
\end{theorem}
\begin{proof}
    Since $\bm{A}$ is invertible,
    \BS
    \begin{align}
        p_A(\lambda) &= |\lambda\bm{I}-\bm{A}| \\
        &= |\bm{A}(\lambda\bm{A}^{-1}-\bm{I})| \\
        &= |\bm{A}|\cdot|\lambda\bm{A}^{-1}-\bm{I}|.
    \end{align}
    \ES
    Note that $\lambda\bm{A}^{-1}-\bm{I}$ is a tridiagonal matrix with entries
    \begin{align}
        c_{i, j} = 
        \begin{cases}
        (\delta_{i-1} + \delta_i)\lambda-1, & \quad i = j \\[1mm]
        -\delta_{\min(i,j)}\lambda, & \quad |i-j| = 1 \\[1mm]
        0, & \quad |i-j| > 1
        \end{cases}.
    \end{align}
    For $k \in \{1, \cdots\!, n\}$, let $f_k$ be the $k$-th leading principal minor of $\lambda\bm{A}^{-1}-\bm{I}$. It is easy to verify that $f_1 = \delta_1\lambda-1$. Then, we can obtain the recurrence in (\ref{Eqn:rec}) by Theorem 2.1 in \cite{el2004inverse}. Since $f_n = |\lambda\bm{A}^{-1}-\bm{I}|$, we have proved the theorem.
\end{proof}

\subsection{Other properties}\label{section:other}
\begin{theorem}
\BS
For any $[h_1, \cdots\!, h_n]^{\rm T} \in \mathbb{R}^n$, let the upper triangular matrix $\bm{H} \equiv [h_{i,j}]_{n \times n}$ be
\begin{align}
    h_{i, j} = 
    \begin{cases}
        h_j & \quad i < j \\[1mm]
        \sum_{k=1}^j h_k & \quad i=j
    \end{cases},
\end{align}
i.e.
\begin{align}
    \bm{H} = 
    \begin{bmatrix}
        h_1 & h_2 & \cdots & h_n \\
        & \sum_{k=1}^2 h_k & \vdots & \vdots \\
        &  & \ddots & h_n \\ 
        & \text{\Large{0}} &  & \sum_{k=1}^n h_k
    \end{bmatrix},
\end{align}
then $\bm{Q} = \bm{H}\bm{A}$ is an L-banded matrix with $[q_1, \cdots\!, q_n]$, where for $t \in \{1, \cdots\!, n\}$,
\begin{align}
    q_t = a_t\sum_{k=1}^t h_k + \sum_{k=i+1}^n h_k a_k.
\end{align}
\ES
\end{theorem}
\begin{proof}
For $i,j \in \{1, \cdots\!, n\}$, let $\bm{h}_i^{\rm T}$ denote the $i$-th row of $\bm{H}$ and $\bm{a}_j$ denote the $j$-th column of $\bm{A}$. We have
\BS
\begin{align}
    \bm{h}_i^{\rm T} &= \big[\underbrace{0, \cdots\!, 0}_{i-1}, \textstyle\sum_{k=1}^i h_k, h_{i+1}, \cdots\!, h_n\big], \\
    \bm{a}_j &= \big[\underbrace{a_j, \cdots\!, a_j}_{j}, a_{j+1}, \cdots\!, a_n\big]^{\rm T}.
\end{align}
Let $t = \max(i, j)$. The $(i, j)$-th entry of $\bm{Q}$ is 
\begin{align}
    q_{i, j} &= \bm{h}_i^{\rm T} \bm{a}_j \\
             &= a_t\sum_{k=1}^t h_k + \sum_{k=i+1}^n h_k a_k = q_t.
\end{align}
\ES
Thus, $\bm{Q}$ is an L-banded matrix with $[q_1, \cdots\!, q_n]$.
\end{proof}

\begin{proposition}
Let $\bm{A}^2 \equiv [b_{i,j}]_{n \times n}$, $t = \max(i, j)$ and $m = \min(i, j)$. Then, 
\begin{align}
    b_{i, j} &= a_t\Big(m a_m + \sum_{k=m+1}^{t}a_k\Big) + \sum_{k=t+1}^n a_k^2.
\end{align}
\end{proposition}
\begin{proof}
For $i,j \in \{1, \cdots\!, n\}$, let $\bm{a}_i^{\rm T}$ be the $i$-th row of $\bm{A}$.
\BS
\begin{align}
    \bm{a}_i^{\rm T} &= [a_i, \cdots\!, a_i, a_{i+1}, \cdots\!, a_n]. \\
    b_{i, j} &= \bm{a}_i^{\rm T} \bm{a}_j \\
             &= m a_m a_t + \sum_{k=m+1}^{t}a_k a_t + \sum_{k=t+1}^n a_k^2.
\end{align}
\ES
Thus, we have proved the proposition.
\end{proof}

\begin{proposition}\label{pro:2}
The linear combination of L-banded matrices is an L-banded matrix. 
\end{proposition}
\begin{proof}
The linear combination of matrices can be seen element-wise. Thus, the linear combination of any number of L-banded matrices is still an L-banded matrix.
\end{proof}

\subsection{New Derivation for Determinant}\label{section:det}
For $i, j \in \{1, \cdots\!, n\}$, let $\bm{A}_{i, j}$ denote the matrix $\bm{A}$ with row $i$ and column $j$ removed. 
\begin{lemma}\label{lemma:1}
    For $n \geq 2$, $k \in \{1, \cdots\!, n\}$,
    \begin{align}
        |\bm{A}_{1, k}| = 
        \begin{cases}
            |\bm{A}_{1, 1}|,  &  \quad 1 \leq k \leq 2 \\[1mm] 
            0,  &  \quad 3 \leq k \leq n  
        \end{cases}.
    \end{align}
\end{lemma}
\begin{proof}
    The column 1 and column 2 of $\bm{A}$ are $[a_1, \cdots\!, a_n]^{\rm T}$ and $[a_2, a_2, \cdots\!, a_n]^{\rm T}$, respectively. Then, column 1 and column 2 of $\bm{A}$ with row 1 removed are both $[a_2, \cdots\!, a_n]^{\rm T}$. Thus, $|\bm{A}_{1, k}| = 0$ for $k \geq 3$ since the column 1 and column 2 of $\bm{A}_{1, k}$ are the same. In addition, $|\bm{A}_{1, 1}| = |\bm{A}_{1, 2}|$.
\end{proof}

Then, we give a new proof of Lemma \ref{lemma:det} by mathematical induction. \\
(1) Base case: For $n = 1$, $\bm{A} = [a_1]$ so that $|\bm{A}| = a_1$. \\
(2) Induction step: For $n = p$, $\widetilde{\bm{A}}$ is a $p \times p$ L-banded matrix with $[a_1, \cdots\!, a_{p}]$ and we assume that Lemma \ref{lemma:det}  holds. For $n = p\!+\!1$, $\bm{A}$ is a $(p\!+\!1) \times (p\!+\!1)$ L-banded matrix with $[a_1, \cdots\!, a_{p+1}]$. The Laplace expansion \cite{golub2013matrix} of $\bm{A}$ is 
\BS
\begin{align}
    |\bm{A}| &= \sum_{j=1}^{p+1} (-1)^{j+1} a_j |\bm{A}_{1, j}| \\
    & = (a_1 - a_2) |\bm{A}_{1, 1}|, \label{Eqn:1}
\end{align}
\ES
where (\ref{Eqn:1}) holds because of Lemma \ref{lemma:1}. Note that $\bm{A}_{1,1}$ is a $p \times p$ L-banded matrix with $[a_2, \cdots\!, a_{p+1}]$. By the assumption for $n=p$, we have  
\begin{align}\label{Eqn:2}
    |\bm{A}_{1, 1}| = a_{p+1}\prod_{k=2}^{p}(a_k - a_{k+1}).
\end{align}
From (\ref{Eqn:1}) and (\ref{Eqn:2}),
\begin{align}
    |\bm{A}| = a_{p+1}\prod_{k=1}^{p}(a_k - a_{k+1}).
\end{align}
Thus, Lemma \ref{lemma:det} holds for $n = p\!+\!1$.

\subsection{New Derivation for Inverse}\label{section:inv}
We give a new proof of Lemma \ref{lemma:inv} by mathematical induction. \\
(1) Base case: For $n = 1$, $\bm{A} = [a_1]$ so that $\bm{A}^{-1} = [a_1^{-1}]$. \\
(2) Induction step: For $n = p$, $\widetilde{\bm{A}}$ is a $p \times p$ L-banded matrix with $[a_1, \cdots\!, a_{p}]$ and we assume that the theorem holds. For $n = p\!+\!1$, $\bm{A}$ is a $(p\!+\!1) \times (p\!+\!1)$ L-banded matrix with $[a_1, \cdots\!, a_{p+1}]$. Since $\widetilde{\bm{A}}$ and $\bm{A}$ are invertible, for $k \in \{1, \cdots\!, p\}$, $a_k \neq a_{k+1}$, $a_{p} \neq 0$ and $a_{p+1} \neq 0$. Let $\bm{1} \in \mathbb{R}^p$ be an all-ones vector. It is easy to verify that
\BS
\begin{align}
    \bm{1}^{\rm T}\widetilde{\bm{A}}^{-1} &= \big[0, \cdots\!, 0, a_p^{-1}\big], \\
    \widetilde{\bm{A}}^{-1} \bm{1} &= (\bm{1}^{\rm T}\widetilde{\bm{A}}^{-1})^{\rm T}, \\
    \bm{1}^{\rm T}\widetilde{\bm{A}}^{-1}\bm{1} &= a_p^{-1}. 
\end{align}
\ES
Then, we let
\BS
\begin{align}
    s &= \big(a_{p+1} - a_{p+1}^2 \bm{1}^{\rm T} \widetilde{\bm{A}}^{-1} \bm{1}\big)^{-1} \\
      &= a_p a_{p+1}^{-1}(a_p - a_{p+1})^{-1} \\
      &= (a_p - a_{p+1})^{-1} + a_{p+1}^{-1}.
\end{align}
\ES
Expressing the matrix inverse in block form,
\BS
\begin{align}
    &\bm{A}^{-1} = 
    \begin{bmatrix} 
        \widetilde{\bm{A}} & a_{p+1}\bm{1} \\
        a_{p+1}\bm{1}^{\rm T} & a_{p+1}
    \end{bmatrix}^{-1} \\
    &=
    \begin{bmatrix} 
        \widetilde{\bm{A}}^{-1}\! +\! a_{p+1}^2 s \widetilde{\bm{A}}^{-1} \bm{1} \bm{1}^{\rm T} \widetilde{\bm{A}}^{-1} & -a_{p+1} s \widetilde{\bm{A}}^{-1} \bm{1} \\
        -a_{p+1} s \bm{1}^{\rm T}\widetilde{\bm{A}}^{-1} & s
    \end{bmatrix}. \label{Eqn:inv_2}
\end{align}
\ES    
Note that
\BS
\label{Eqn:inv_3}
\begin{align}
    a_{p+1}^2 s \widetilde{\bm{A}}^{-1} \bm{1} \bm{1}^{\rm T} \widetilde{\bm{A}}^{-1} &= 
    \begin{bmatrix}
        \ \bm{0} & \bm{0} \\
        \ \bm{0} & (a_p\!-\!a_{p+1})^{-1}\!-\!a_p^{-1}
    \end{bmatrix}, \\
    -a_{p+1} s \bm{1}^{\rm T}\widetilde{\bm{A}}^{-1} &= \big[0, \cdots\!, 0, (a_p-a_{p+1})^{-1}\big].
\end{align}
\ES
For $k \in \{1, \cdots\!, p\}$, let $\delta_k = (a_k - a_{k+1})^{-1}$, $\delta_0 = 0$ and $\delta_{p+1} = a_{p+1}^{-1}$. Then, substitute (\ref{Eqn:inv_3}) into (\ref{Eqn:inv_2}), we can find that $\bm{A}^{-1}$ is the same as the matrix in (\ref{Eqn:inv_1}). Thus, the  holds for $n = p\!+\!1$.

\begin{subsection}{Complexity Comparison}\label{section:comp}
We give a comparison of the time complexity of some operations on L-banded matrices and those of general real symmetric matrices. 
\begin{table}[h!] 
\renewcommand{\arraystretch}{1.5} 
\centering 
\footnotesize 
\setlength{\tabcolsep}{1mm}{
\begin{tabular}{c|c|c}
\hline
 & General matrices & L-banded matrices \\
\hline
Determinant & ${\cal O}(n^3)$ & ${\cal O}(n)$  \\
\hline
Inverse & ${\cal O}(n^3)\;\star$ & ${\cal O}(n)\;\star$ \\
\hline
$\bm{x}^{\rm T}\bm{A}\bm{x}$ & ${\cal O}(n^2)$ & ${\cal O}(n)$ \\
\hline
Definiteness & ${\cal O}(n^3)$ & ${\cal O}(n)$ \\
\hline
LDL decomposition & ${\cal O}(n^3)$ & ${\cal O}(n^2)$ \\
\hline
Cholesky decomposition & ${\cal O}(n^3)\;\dag$ & ${\cal O}(n^2)\;\dag$ \\
\hline
Minors / Cofactors & ${\cal O}(n^3)$ & ${\cal O}(n)\;\star$ \\
\hline
\makecell{Determinant after a\\ column substitution} & ${\cal O}(n^3)$ & ${\cal O}(n)\;\star$ \\
\hline
Characteristic polynomial & ${\cal O}(n^3)$ & ${\cal O}(n)\;\star$ \\
\hline
\end{tabular}}
\vspace{\baselineskip}
\caption{Comparison for the complexity of operations on L-banded matrices and general symmetric matrices}
\label{table:comp}
\end{table}

In Table \ref{table:comp}, the notation ``$\star$'' and ``$\dag$'' means the corresponding matrices are invertible and positive definite, respectively. In addition, we need to point out that some methods can find determinants of general matrices with complexity between ${\cal O}(n^2)$ and ${\cal O}(n^3)$. However, the most common methods, like LU decomposition or Bareiss algorithm, are in ${\cal O}(n^3)$. 
\end{subsection}

\end{section}

\begin{section}{Conclusions}
In this paper, we gave many algebraic properties of L-banded matrices. We expect that our findings in this research will contribute to the fields of mathematics, iterative signal processing, message-passing algorithms, and other relevant applications that employ the L-banded matrices.
\end{section}

\bibliographystyle{IEEEtran}
\bibliography{reference}

\end{document}